\documentclass[11pt]{article}

\usepackage[margin=1.25in]{geometry}
\usepackage{amsmath,amssymb,amsthm}
\usepackage{booktabs}
\usepackage{array}
\newcolumntype{P}[1]{>{\raggedright\arraybackslash}p{#1}}
\usepackage{microtype}
\usepackage{setspace}
\usepackage{enumitem}
\usepackage{hyperref}
\usepackage{xcolor}
\usepackage{graphicx}

\hypersetup{colorlinks=true,linkcolor=black,citecolor=black,urlcolor=blue}
\newtheorem{proposition}{Proposition}

\newtheorem{lemma}{Lemma}
\theoremstyle{definition}

\theoremstyle{remark}

\newcommand{\ind}{\mathbf{1}}

\title{\textbf{Organizational Lifespan as Commitment:\\ The Case of Foundations}}
\author{Christian Jaag\thanks{Swiss Economics, Ottikerstrasse 7, 8006 Z\"urich, Switzerland and Somaha Foundation, Weinbergstrasse 102, 8006 Z\"urich, Switzerland. Corresponding author: \texttt{christian.jaag@somaha-stiftung.ch}.}}
\date{October 2026}

\begin{document}
\maketitle

\begin{abstract}
An organization may bind future decision makers through a commitment over its own lifespan. This paper studies that choice for philanthropic foundations. A foundation may commit to terminate, commit to institutional continuity, or leave termination to future discretion. Discretion preserves state-contingent stopping. Perpetuity instead raises investment in foundation-specific organizational capital whose return is partly lost upon termination. A terminal commitment limits excessive continuation when future governing bodies obtain private benefits from organizational survival. A canonical model yields nonempty regions in which each regime is optimal, while a finite-horizon extension endogenizes the terminal date. Perpetuity is optimal when organizational capital is sufficiently productive relative to the cost of continuing in bad states; discretion is optimal when stopping flexibility is valuable and governance is aligned; a finite lifespan is optimal when survival incentives create sufficiently costly entrenchment. Organizational lifespan is therefore a distinct commitment margin.
\end{abstract}


\section{Introduction}

Organizations routinely constrain future decision makers. They allocate control rights, write long-term contracts, protect particular objectives, and restrict the use of assets. A more fundamental commitment receives less attention: an organization may also restrict future discretion over whether the organization itself continues to exist. Committing to survive protects investments whose returns depend on institutional continuity. Committing to terminate constrains future decision makers who may value organizational survival too highly. Leaving lifespan open preserves the option to respond to information that arrives later. This paper studies the trade-off among these three forms of institutional design.

Philanthropic foundations provide a particularly clear setting. A foundation may be established to spend down by a specified date, to exist in perpetuity, or without a predetermined terminal date. These arrangements govern the lifespan of the organization rather than merely its rate of spending. They matter because a foundation is both a pool of assets and an organization. It may accumulate expertise, screening capacity, reputation, networks, and institutional memory whose value depends on continued operation. It may also develop a bureaucracy whose governing bodies obtain benefits from preserving the institution after its social value has declined.

The analysis compares three regimes. A spend-down foundation $S(T)$ commits to terminate at a finite date $T$. A perpetual foundation $P$ commits not to terminate voluntarily while solvent. An open-ended foundation $N$ leaves the stopping decision to future governing bodies. The distinction is a commitment choice over future control of the institution's existence: $P$ removes discretion to terminate, $S(T)$ removes discretion to continue beyond $T$, and $N$ preserves both options.

The starting point is a flexibility benchmark. If the lifespan rule does not affect current behavior or future governance, $N$ weakly dominates either commitment. Future decision makers observe information unavailable at formation and may implement the same continuation path that would have been imposed ex ante. A restriction on future stopping therefore requires an incentive or governance benefit to compensate for the lost option value.

The first benefit arises from investment in organization-specific capital. Before future conditions are known, the foundation invests in capabilities that are productive only while the institution remains active. Under $N$, later decision makers choose whether to stop after the investment is sunk. Their ex-post stopping rule takes the installed organizational capital as given and therefore does not internalize how the prospect of termination weakened the earlier investment incentive. Perpetuity removes this ex-ante commitment problem and raises investment. It may therefore be valuable even though stopping discretion is desirable once the state is known. In the canonical model, the welfare gain from $P$ rises with the productivity of organizational capital, whereas its cost is continuation in states in which the institution has negative direct social value. The model yields a closed-form threshold at which the investment effect dominates the stopping loss.

The second benefit works in the opposite direction. Future governing bodies may obtain private benefits from organizational survival through employment, prestige, control, or attachment to the institution. Sequential discretion then creates an entrenchment problem: the institution may continue after its social continuation value has become negative. A terminal commitment limits that distortion. A finite-horizon extension shows that the optimal terminal date need not be immediate. It balances the direct social value of another period and the additional return to organizational investment against the governance cost of allowing later continuation.

These two commitment problems generate a three-way institutional choice. There are nonempty parameter regions in which each lifespan regime is optimal. Perpetuity is preferred when durable organization-specific capital is sufficiently productive. Open-ended discretion is preferred when the option value of state-contingent stopping is large and governance is aligned. A finite lifespan is preferred when future survival incentives generate sufficiently costly overcontinuation. The result gives the three institutional forms symmetric interpretations: commitment to continuity protects organizational investment, commitment to termination disciplines future governance, and discretion preserves adaptation.

The mechanism is distinct from familiar financing margins. Institutional lifespan differs from payout because a long-lived foundation may spend principal aggressively, while a finite-life foundation may initially retain substantial assets. It differs from grant duration because an enforceable multiyear contract may bind support to a recipient without binding the foundation to exist indefinitely. It also differs from real-corpus preservation, which is an additional constraint on the asset path rather than a necessary implication of perpetuity. Separating these margins isolates the organizational commitment problem.

The paper contributes to organizational economics by adding institutional survival to the set of choices through which current principals shape future discretion. Standard incomplete-contracting arguments emphasize control rights over future actions and assets conditional on the organization remaining in place. Here existence itself is the committed decision margin: the lifespan rule determines whether future controllers retain an organization within which those rights may be exercised. Foundations make this margin unusually observable because their charters may explicitly choose perpetuity or a terminal date. The same trade-off arises in other mission-oriented organizations with durable mandates, including research institutes, endowed cultural institutions, and nonprofit operating organizations, whenever organizational capital is partly nontransferable and later controllers may value survival differently from the original principal.

The paper also contributes to the economics of foundations. Existing work studies payout, endowment management, multiyear grants, and limited-life foundations. The present analysis makes the foundation's own lifespan the choice variable and formally includes the open-ended regime that lies between predetermined termination and perpetuity. The central comparison therefore concerns investment incentives, governance, and stopping discretion rather than alternative spending rates.

Section 2 positions the analysis relative to organizational commitment and the foundation literature. Section 3 defines the institutional choice and the flexibility benchmark. Sections 4 and 5 study commitment to continuity and commitment to termination. Section 6 synthesizes the optimal lifespan. Section 7 considers information, grant relationships, and financing constraints. Section 8 discusses institutional design, and Section 9 concludes. Proofs are collected in the Appendix.

\section{Related literature}

The paper first relates to organizational economics and incomplete contracting. Durable organizational arrangements allocate authority over future actions when relevant contingencies cannot be fully specified ex ante (Grossman and Hart, 1986; Hart and Moore, 1990). Hansmann and Thomsen (2021) show, in the setting of foundation-owned firms, how the foundation form may provide long-run commitment. The present paper studies a different commitment margin. Rather than allocating control over an action conditional on organizational survival, the lifespan rule makes existence itself subject to ex-ante commitment. This distinction creates a specific dynamic problem: an ex-post efficient stopping rule may weaken ex-ante investment in nontransferable organizational capital, while unrestricted future control may also permit entrenchment. The model therefore treats continuation of the institution as an endogenous element of organizational design.

The paper is related second to the literature on foundation payout and endowment management. Deep and Frumkin (2006) describe the foundation payout decision as a choice between current and future grantmaking. Klausner (2003) explains why financial compounding alone does not resolve the intergenerational allocation problem, because the relevant comparison includes the social return to charitable spending. Afik and Katz (2019) show how payout and portfolio choices affect longevity under alternative capital-market environments. Hansmann (1990) studies the functions of long-lived endowments, while Brown et al. (2014) and Roth Tran (2019) examine smoothing and mission-related portfolio choices. Lo, Matveyev, and Zeume (2025, revised 2026) provide large-sample evidence on nonprofit endowments, including their relation to mission spending, revenue hedging, investment performance, and governance. The present paper separates this asset-allocation problem from the organizational choice of lifespan.

A related recent contribution is Matveyev and Zeume (2026), who study nonprofit liquidity and donor behavior. Their cash-rich paradox qualifies a mechanical insurance interpretation of retained assets because reported reserves may affect external support. This channel matters for the financial value of continuing assets but is separate from the organizational-capital and governance mechanisms in the main model.

A third literature studies limited-life foundations directly. Thelin and Trollinger (2009) document historical and contemporary spend-down cases. Loh and Buteau (2017) show that limited-life foundations make distinct choices about strategy, staffing, grantmaking, evaluation, and grantee relationships. Hengevoss and von Schnurbein (2025) identify strategic patterns among sunset foundations. Recent practice has also questioned the binary framing of perpetuity versus spend-down. Mananzala and Rojas (2026) describe a ``tidal'' approach in which spending and corpus evolve with mission conditions rather than a fixed terminal date. These contributions motivate treating open-ended discretion as a separate institutional regime. The model asks when either direction of lifespan commitment dominates that benchmark.

The analysis also relates to multiyear grantmaking. Afik, Levy, and Katz (2018) study payout jointly with multiyear grants and show that grant duration is a separate margin of foundation policy. This distinction is important here: when enforceable grants span every payoff-relevant date, institutional perpetuity adds no grantee-commitment value. The main model therefore does not rely on recipient investment. It instead studies investment by the foundation in its own organization-specific capital.

Recent work has renewed the normative debate over perpetuity. Ochoa-Carreno (2026) studies perpetuity and spend-down in relation to investment and technological transitions. Ghadamian (2026) argues that perpetual endowments may misallocate capital when real-corpus preservation constrains mission spending. The present paper separates institutional continuity from corpus preservation. A foundation may commit to remain active without requiring the real value of its initial endowment to remain constant.

Related analyses study adjacent commitment and financing margins. Jaag (2026a) studies grants versus repayable impact finance at a fixed gross commitment. Jaag (2026b) compares private, philanthropic, and state capital in staged financing, where provider identity determines later continuation and recipient relationship-specific investment. The foundation form is fixed here, and the object of commitment is the foundation's lifespan; recipient investment appears only in an extension. Jaag (2026c) studies mission commitments chosen before future intervention contracts and focuses on organizational rights and responsiveness to later contracting. The analysis here takes the mission as given and studies discretion over institutional survival. Financing instruments, provider assignment, mission design, and lifespan design are distinct margins.

\section{Institutional choice}

\subsection{Timing and lifespan regimes}

The foundation chooses its lifespan rule before making irreversible investments and before future conditions are known. The sequence is: at date 0 the founder selects a lifespan regime; the foundation then invests in organization-specific capital; future social conditions are realized; and, when the regime leaves discretion, future governing bodies decide whether the organization continues. This timing isolates the commitment problem. A later stopping decision treats earlier investment as sunk, whereas the founder anticipates how the stopping rule changes the return to that investment.

The analysis compares three regimes. A perpetual foundation $P$ commits against voluntary termination while continued operation remains feasible. A spend-down foundation $S(T)$ commits to terminate after a finite number $T\geq1$ of future operating periods. An open-ended foundation $N$ specifies no terminal date and leaves future continuation to its governing bodies. Perpetuity concerns institutional continuity rather than preservation of the initial corpus; spend-down concerns a terminal date rather than a particular payout path.

Let $C^{soc}$ denote the social continuation value of an active organization at a future decision date, excluding private benefits of its controllers. Under aligned governance, a discretionary board continues exactly when $C^{soc}\geq0$. Future controllers may instead obtain a private survival benefit $d\geq0$ from employment, prestige, control, or attachment to the institution. Their continuation rule then becomes
\begin{equation}
C^{soc}+d\geq0.
\label{eq:generalboardrule}
\end{equation}
The parameter $d$ is a reduced-form governance wedge: it affects future continuation but is not part of social welfare.

\subsection{The flexibility benchmark}

\begin{lemma}[Flexibility benchmark]\label{lem:flexibility}
If governance is aligned and the lifespan rule does not affect earlier investment or other behavior, open-ended discretion weakly dominates either a fixed terminal date or perpetuity. The dominance is strict whenever the committed regime requires a stop-versus-continue action that differs from the socially optimal action with positive probability.
\end{lemma}

The lemma identifies what commitment must accomplish. Without an incentive or governance effect, $N$ retains every stop-versus-continue action available under a fixed rule and adds the ability to condition that action on later information. Perpetuity and spend-down therefore require distinct commitment benefits to compensate for lost flexibility.

\section{Commitment to continuity}

Consider a two-state specialization. At date 0 the foundation chooses organization-specific investment $e\geq0$ at cost
\begin{equation}
\frac{c}{2}e^2,\qquad c>0.
\label{eq:canoncost}
\end{equation}
If the institution operates in the future period, the investment adds $\eta e$ to social value, where $\eta\geq0$ measures the productivity and institution specificity of organizational capital. If the institution terminates, that payoff is lost. The future direct social continuation value is
\begin{equation}
z\in\{z_H,z_L\},\qquad z_H>0>z_L,
\end{equation}
with $\Pr(z=z_H)=p\in(0,1)$. The lifespan rule is chosen before $e$, and the state is observed only after $e$ is sunk.

First suppose governance is aligned, $d=0$, and open-ended discretion continues in the high state and stops in the low state. Its date-0 welfare is
\begin{equation}
W_N(e)=p(z_H+\eta e)-\frac{c}{2}e^2,
\label{eq:WNcanonical}
\end{equation}
so
\begin{equation}
e_N^*=\frac{p\eta}{c},
\qquad
W_N^*=pz_H+\frac{p^2\eta^2}{2c}.
\label{eq:WNsolution}
\end{equation}
The low-state stopping condition is
\begin{equation}
z_L+\frac{p\eta^2}{c}<0.
\label{eq:lowstopcondition}
\end{equation}

Perpetuity operates in both states. Hence
\begin{equation}
W_P(e)=p z_H+(1-p)z_L+\eta e-\frac{c}{2}e^2,
\label{eq:WPcanonical}
\end{equation}
with
\begin{equation}
e_P^*=\frac{\eta}{c},
\qquad
W_P^*=p z_H+(1-p)z_L+\frac{\eta^2}{2c}.
\label{eq:WPsolution}
\end{equation}
Thus $e_P^*>e_N^*$ whenever $p<1$ and $\eta>0$. Optimized welfare differs by
\begin{equation}
W_P^*-W_N^*=(1-p)\left[z_L+\frac{(1+p)\eta^2}{2c}\right].
\label{eq:perpdiscgap}
\end{equation}
Define
\begin{equation}
\eta_P\equiv\sqrt{\frac{-2c z_L}{1+p}}.
\label{eq:etathreshold}
\end{equation}

\begin{proposition}[Continuity commitment]\label{prop:continuity}
Under aligned governance and condition (\ref{eq:lowstopcondition}), open-ended discretion dominates perpetuity when $\eta<\eta_P$, while perpetuity dominates open-ended discretion when $\eta>\eta_P$.
\end{proposition}

The advantage of perpetuity is an ex-ante commitment effect. Under $N$, the future board decides after $e$ is sunk and therefore does not internalize how the prospect of later termination reduced the earlier return to organization-specific investment. Perpetuity removes that stopping option and raises investment. Its benefit must exceed the social loss from operating in the low state. The result therefore does not follow mechanically from a longer horizon.

\section{Commitment to termination}

Now consider future operating periods $t=1,2,\ldots$. A terminal commitment $S(T)$ permits operation through period $T$ and requires termination thereafter. Suppose a sufficiently strong survival bias would otherwise keep the institution active whenever continued operation remains feasible. Let $b_t$ denote the date-0 social value of operating in period $t$ before the return to date-0 organization-specific investment. The investment contributes $\eta q^{t-1}e$ in period $t$, where $q\in(0,1)$ measures persistence. Social welfare under $S(T)$ is
\begin{equation}
W_{S(T)}(e)=\sum_{t=1}^{T}\beta^{t-1}\left(b_t+\eta q^{t-1}e\right)-\frac{c}{2}e^2.
\label{eq:WST}
\end{equation}
Define
\begin{equation}
B_T\equiv\sum_{t=1}^{T}(\beta q)^{t-1}=\frac{1-(\beta q)^T}{1-\beta q}.
\label{eq:BT}
\end{equation}
Then
\begin{equation}
e_T^*=\frac{\eta B_T}{c},
\qquad
W_{S(T)}^*=\sum_{t=1}^{T}\beta^{t-1}b_t+\frac{\eta^2 B_T^2}{2c}.
\label{eq:WSTstar}
\end{equation}
The marginal value of extending the terminal date from $T$ to $T+1$ is
\begin{equation}
\Delta_T=\beta^T b_{T+1}+\frac{\eta^2}{2c}\left[2B_T(\beta q)^T+(\beta q)^{2T}\right].
\label{eq:deltaT}
\end{equation}
The first term is the direct social value of another operating period. The second is the investment effect of a longer credible horizon.

\begin{proposition}[Optimal finite terminal date]\label{prop:optimalT}
Suppose $b_1>0$, the sequence $\{b_t\}$ is eventually weakly decreasing with $\limsup_{t\to\infty}b_t<0$, and $q\in(0,1)$. Then $\Delta_T<0$ for all sufficiently large $T$. If $\Delta_T$ single-crosses from positive to negative, there is a unique finite optimal terminal date $T^*\geq1$ characterized by
\begin{equation}
\Delta_{T^*-1}\geq0,
\qquad
\Delta_{T^*}<0.
\label{eq:Tstarcondition}
\end{equation}
If survival bias makes open-ended discretion continue indefinitely, $S(T^*)$ strictly dominates biased discretion and perpetuity whenever the single crossing is strict.
\end{proposition}

A finite lifespan therefore need not mean rapid liquidation. Stronger or more persistent organization-specific capital raises the investment value of extending the horizon and tends to delay $T^*$. A social continuation profile that deteriorates more rapidly brings the terminal date forward. The relevant comparison is between the social value of continued institutional operation, its effect on earlier investment, and the governance distortion from leaving continuation to future controllers.

\section{The optimal lifespan}

The two preceding commitment problems yield the answer to the institutional choice.

\begin{proposition}[Three lifespan regimes]\label{prop:three}
There are nonempty parameter regions in which each lifespan class is optimal. With aligned governance, open-ended discretion is preferred when the value of state-contingent stopping exceeds the investment benefit of continuity commitment, while perpetuity is preferred when organization-specific capital is sufficiently productive. With sufficiently strong survival bias and an eventually negative social continuation profile, a finite terminal commitment $S(T^*)$ is preferred when Proposition \ref{prop:optimalT} applies and its strict single crossing makes the optimal finite path dominate indefinite continuation.
\end{proposition}

The three regimes therefore address different problems. Perpetuity commits to continuity in order to protect irreversible organizational investment. Open-ended discretion preserves adaptation when future information matters and governance is aligned. Spend-down commits to termination in order to prevent future controllers from preserving an organization beyond its socially valuable life. Financial returns and payout policies may affect the underlying continuation values and financing feasibility, but neither alone determines which lifespan regime is optimal.

\section{Extensions}

\subsection{Information and external learning}

The preceding institutional-capital mechanism is an investment effect. External learning generates the opposite force because information received after date 0 increases the value of preserving a stopping option.

Suppose an additional signal arrives before a future continuation decision and refines beliefs about $C_t^{soc}$ or future grant opportunities. A discretionary foundation may use the signal when deciding whether to stop. A spend-down foundation cannot continue beyond $T$, while a perpetual foundation cannot voluntarily terminate.

\begin{proposition}[External learning]
Holding current investment and governance incentives fixed, a more informative signal about future social continuation values weakly increases the value of sequential discretion relative to either fixed lifespan commitment. The gain is strict when the signal changes the socially optimal stop-versus-continue action with positive probability.
\end{proposition}

The result distinguishes two objects often grouped under ``learning.'' Information about the external environment favors flexibility. Investment in foundation-specific knowledge may favor commitment to a long institutional horizon. A foundation that expects rapid change in social needs may therefore value an open stopping rule even while investing heavily in transferable analytical capacity; a foundation whose effectiveness depends on persistent, nontransferable field expertise may place greater value on continuity.

\subsection{Grant relationships}

Institutional lifespan may also affect grantees, but this channel is conceptually separate from the main model. It matters only to the extent that a grant contract cannot itself provide the relevant commitment.

Let $H\geq1$ denote the horizon of an enforceable grant contract with a core grantee and $L\geq H$ the payoff horizon of a relationship-specific grantee investment $x$. The investment produces productive capacity
\begin{equation}
x_t=\varphi^t x\,\ind\{t<L\},
\qquad \varphi\in(0,1],
\label{eq:partnercapital}
\end{equation}
and costs $c(x)$, with $c'>0$ and $c''>0$. For $t<H$, complementary foundation support is contractually guaranteed. For $t\geq H$, support cannot be fully specified ex ante. If the institution remains active, its future governing body chooses whether to renew the relationship. Let $a_t\in\{0,1\}$ denote renewal and let $\Delta_t(x_t,\omega_t)$ be the incremental social surplus from renewal relative to the best alternative use. An aligned board renews when
\begin{equation}
a_t=1\quad\Longleftrightarrow\quad \Delta_t(x_t,\omega_t)\geq0.
\label{eq:renewalrule}
\end{equation}
Perpetuity therefore guarantees neither a grant nor renewal. It guarantees only that the institution remains available to make the future allocation decision.

\begin{proposition}[Grant contracts and institutional lifespan]
If $H=L$ and a foundation can honor the contract through $L-1$, all payoff-relevant periods of the grantee's investment are contractually protected. Institutional perpetuity then has no additional commitment value through the grantee-investment channel. If $H<L$, institutional continuity may raise expected renewal only in states in which the relationship would be socially retained conditional on the foundation remaining active but the open-ended institution would otherwise terminate.
\end{proposition}

When the grantee captures only part of the social return to $x$, the post-contract renewal problem may generate underinvestment. A longer institutional horizon may mitigate this hold-up if it preserves payoff-relevant renewal states. This mechanism is deliberately secondary here. A separate provider-assignment problem asks how different capital providers' continuation rules affect recipient investment. The present paper instead uses the extension only to show why lifespan and grant duration should not be conflated.

\subsection{Risk, external funding, and capital preservation}

Risk changes the value of retaining a long institutional horizon but does not alter the commitment logic above. Let the real return become stochastic, $R_{t+1}$, so that
\begin{equation}
A_{t+1}=(A_t-f-g_t-e_t)(1+R_{t+1}).
\label{eq:riskbudget}
\end{equation}
Social productivity $\theta_t$ may also be stochastic and correlated with returns. Retained assets then finance future grants and organizational investment and may provide resources in states in which mission needs are high. The insurance value depends on the joint distribution of returns, mission opportunities, and other financing sources rather than on return variance alone.

Recent evidence also cautions against treating reserves as a purely internal hedge. Matveyev and Zeume (2026) show that nonprofit liquidity may affect donor behavior, so a larger balance sheet need not translate one-for-one into greater resources after a funding shock. In reduced form, let expected external support be $D(A_t)$, with $D'(A_t)$ potentially negative when outside donors interpret reserves as lower need. The marginal value of retaining an additional foundation dollar then includes both its financial payoff and any induced change in future external funding. This channel does not determine lifespan by itself, but it makes the insurance value of endowment wealth endogenous to the broader financing environment.

A perpetual charter should not be interpreted as a guarantee of physical survival in every financial state. Under risk, sufficiently adverse returns may make continued operation infeasible. The commitment in $P$ is therefore a commitment against voluntary termination while solvent. This interpretation avoids requiring finite initial wealth to finance positive operating costs almost surely forever.

The distinction between institutional perpetuity and capital preservation remains important. Consider an additional asset-floor rule
\begin{equation}
A_t\geq \underline A_t
\qquad \text{for all }t,
\label{eq:floor}
\end{equation}
where $\underline A_t$ is externally specified. Preserving the real value of the initial corpus is one possible specification. The perpetual charter itself imposes no such floor.

\begin{proposition}[Capital preservation]
Holding institutional-investment and governance effects fixed, adding a binding asset-floor constraint to a perpetual foundation weakly lowers maximum social welfare. The loss is strict when the unconstrained perpetual policy would use assets below the floor in a state in which their marginal social deployment value is positive.
\end{proposition}

This proposition does not imply that an asset floor is always undesirable. A floor may itself address a governance problem or protect promised long-run activity. The narrower point is that the economic case for institutional perpetuity and the economic case for preserving the initial corpus are separate.

With $f>0$, even a perpetual foundation without a preservation mandate needs enough financial capacity to remain operational. In a deterministic zero-grant, zero-investment stationary path, the minimum stationary asset level satisfies
\begin{equation}
\bar A=(\bar A-f)(1+r),
\qquad
\bar A=\frac{1+r}{r}f.
\label{eq:continuationcapacity}
\end{equation}
Thus perpetuity permits use of part of the initial corpus whenever $A_0>\bar A$, but it cannot consume the continuation capacity required by its own institutional commitment. Operational viability remains distinct from preservation of $A_0$.

\section{Policy and institutional design}

The model has three implications for institutional design. First, payout, lifespan, and capital preservation are separate margins. A payout rule constrains current deployment relative to assets. A terminal rule constrains how long the institution may exist. A preservation rule constrains the asset path. Bundling these choices obscures the mechanism that each rule is meant to address.

Second, an open-ended legal form has a natural role as a flexibility benchmark. Recent practitioner proposals that reject a fixed choice between perpetuity and spend-down are consistent with this logic. Yet discretion is not generally first best once current institutional investment or future governance incentives respond to the lifespan rule. A credible perpetual horizon may support investment in organizational capacity, while a terminal date may discipline later survival incentives.

Third, a uniform payout rule cannot generally implement the socially optimal policy. In the United States, private nonoperating foundations are generally required to make qualifying distributions based on a distributable amount tied to a minimum investment return, which is generally five percent of relevant noncharitable-use assets subject to statutory adjustments (Internal Revenue Service, 2026). Such a rule may limit indefinite accumulation, but the optimal deployment path in the model varies with financial returns, current social opportunities, institutional-capital investment, future information, and the financing response of other donors.

A practical implication is that a foundation charter should distinguish at least four decisions: the institutional lifespan, the payout policy, the degree of corpus preservation, and the horizon of individual grant commitments. A perpetual institution may rationally adopt a high payout when current opportunities are unusually productive. A finite-life institution may hold substantial reserves early in its life. An open-ended institution may invest in transferable capacity while retaining a state-contingent stopping rule. The institutional form and the spending rule need not move together.

\section{Conclusion}

Organizational lifespan is a commitment margin. Current decision makers may preserve future discretion over survival, commit the organization to continue, or commit it to terminate. These choices trade off adaptation, investment incentives, and future governance. Philanthropic foundations make the trade-off explicit because their governing documents may specify perpetuity or a finite terminal date.

With aligned governance and no behavioral response, open-ended discretion is the benchmark: future decision makers may condition stopping on information unavailable at formation. Perpetuity becomes valuable when the organization accumulates expertise, reputation, networks, screening capacity, or institutional memory whose return depends on continued operation. A credible long horizon then raises investment in organization-specific capital. Its cost is the loss of the option to terminate in states in which continuation has negative social value.

A finite lifespan addresses the opposite commitment problem. Future governing bodies may obtain private benefits from institutional survival and continue after the organization's social value has declined. A terminal date limits that entrenchment but sacrifices valuable continuation states and shortens the horizon over which organizational investment pays. The optimal finite date therefore balances governance discipline against both the direct and investment value of another period.

The resulting three-regime characterization gives each institutional form a distinct economic role. Perpetuity protects investment in durable organizational capacity. A terminal commitment disciplines excessive future survival. Open-ended discretion preserves state-contingent adaptation. None of these roles is equivalent to a payout rate, the duration of a grant contract, or preservation of the initial corpus.

The foundation application also illustrates a broader organizational-design problem. Rules governing future actions matter only while the organization exists. A lifespan commitment determines whether future decision makers retain an organization within which those rules may be exercised at all. The resulting commitment problem has two directions: future stopping discretion may be efficient ex post yet weaken irreversible organizational investment ex ante, while future continuation discretion may preserve an institution after its social value has fallen below zero. Whenever organizational capital is partly nontransferable or future controllers value survival differently from the original principal, the choice between continuity, termination, and discretion may therefore be economically consequential.

\section*{Declaration of generative AI use}
During the preparation of this work the author used ChatGPT (OpenAI) to review theoretical
derivations, refine manuscript structure, and improve readability. The author independently
verified the derivations and reviewed and edited the AI-assisted output as needed, and takes full
responsibility for the content of the published article.

\appendix
\section{Proofs}

\begin{proof}[Proof of Lemma \ref{lem:flexibility}]
With aligned governance and no behavioral response to the lifespan rule, open-ended discretion may reproduce any feasible fixed stopping decision and may also continue whenever continuation is socially optimal. Fixed lifespan rules therefore restrict the state-contingent action set. The weak dominance follows, with strictness whenever a commitment requires an action different from the socially optimal discretionary action with positive probability.
\end{proof}

\begin{proof}[Proof of Proposition \ref{prop:continuity}]
Maximizing (\ref{eq:WNcanonical}) gives $e_N^*=p\eta/c$ and (\ref{eq:WNsolution}). Maximizing (\ref{eq:WPcanonical}) gives $e_P^*=\eta/c$ and (\ref{eq:WPsolution}). Their difference is (\ref{eq:perpdiscgap}), which is positive exactly when $\eta>\eta_P$ and negative when $\eta<\eta_P$.
\end{proof}

\begin{proof}[Proof of Proposition \ref{prop:optimalT}]
From (\ref{eq:deltaT}), the investment term is positive but converges to zero at rate $(\beta q)^T$, while $B_T$ is bounded above by $(1-\beta q)^{-1}$. Because $\limsup_{t\to\infty} b_t<0$, there exist $\varepsilon>0$ and $\bar T$ such that $b_{T+1}\leq-\varepsilon$ for all $T\geq\bar T$. For such $T$, the direct term is at most $-\varepsilon\beta^T$. Dividing the positive investment term by $\beta^T$ leaves terms proportional to $q^T$ and $\beta^Tq^{2T}$, both of which converge to zero. Hence $\Delta_T<0$ for all sufficiently large $T$. Since $b_1>0$, extending the horizon from zero to one period has positive direct value and a positive investment effect, so a positive horizon is preferred to immediate termination. Single crossing then gives the unique $T^*$ in (\ref{eq:Tstarcondition}). Under the final condition in the proposition, the optimized finite path has strictly higher welfare than the indefinite path implemented by biased discretion and perpetuity.
\end{proof}

\begin{proof}[Proof of Proposition \ref{prop:three}]
The aligned-governance regions follow directly from Proposition \ref{prop:continuity}: sufficiently low $\eta$ favors the stopping option under $N$, while sufficiently high $\eta$ favors the investment commitment under $P$. These regions are nonempty because $z_L<0$ and the threshold $\eta_P$ is finite and positive. Proposition \ref{prop:optimalT} supplies a nonempty finite-lifespan region whenever survival bias induces indefinite continuation and the strict single crossing makes the optimized finite path dominate it. The three cases correspond to distinct restrictions on future continuation and therefore establish the three lifespan classes stated in the proposition.
\end{proof}

\begin{proof}[Proof of the external-learning proposition]
A more informative signal expands the information available to a discretionary stopping decision. The decision maker may ignore the signal, so its value is weakly nonnegative. Fixed lifespan commitments cannot use the signal to reverse the committed stop-versus-continue action. The gain is strict when the signal changes the socially optimal action with positive probability.
\end{proof}

\begin{proof}[Proof of the grant-contract proposition]
When $H=L$, every period in which $x_t$ is positive is covered by the enforceable grant contract, so institutional lifespan does not change the supported payoff horizon provided the contract can be honored. When $H<L$, post-contract support requires both institutional activity and $\Delta_t\geq0$. Perpetuity affects expected renewal only through the institutional-activity event. Hence it adds commitment value only in states in which the relationship would be renewed conditional on activity but the open-ended institution would otherwise terminate.
\end{proof}

\begin{proof}[Proof of the capital-preservation proposition]
The asset floor in (\ref{eq:floor}) adds constraints to the perpetual foundation's allocation problem. Removing the floor weakly expands the feasible set, so the constrained optimum cannot exceed the unconstrained optimum. The loss is strict when the unconstrained allocation violates the floor and the resources released by doing so have strictly positive marginal social value.
\end{proof}

\section*{References}
\addcontentsline{toc}{section}{References}
\sloppy

Afik, Z., and H. Katz. 2019. ``Reconsidering the Philanthropic Foundation Minimum Payout Policy under a `New Normal'.'' \textit{Journal of Policy Modeling} 41(2): 219--233. \url{https://doi.org/10.1016/j.jpolmod.2018.09.004}.

Afik, Z., A. Levy, and H. Katz. 2018. ``Philanthropic Foundations Payout and Multiyear Grants: Between Giving Today and Giving Tomorrow.'' \textit{Journal of Wealth Management} 20(4): 33--45. \url{https://doi.org/10.3905/jwm.2018.20.4.033}.

Grossman, S. J., and O. D. Hart. 1986. ``The Costs and Benefits of Ownership: A Theory of Vertical and Lateral Integration.'' \textit{Journal of Political Economy} 94(4): 691--719. \url{https://doi.org/10.1086/261404}.

Hart, O., and J. Moore. 1990. ``Property Rights and the Nature of the Firm.'' \textit{Journal of Political Economy} 98(6): 1119--1158. \url{https://doi.org/10.1086/261729}.

Hansmann, H. 1990. ``Why Do Universities Have Endowments?'' \textit{Journal of Legal Studies} 19(1): 3--42. \url{https://doi.org/10.1086/467841}.

Hansmann, H., and S. Thomsen. 2021. ``The Governance of Foundation-Owned Firms.'' \textit{Journal of Legal Analysis} 13(1): 172--230. \url{https://doi.org/10.1093/jla/laaa005}.

Brown, J. R., S. G. Dimmock, J.-K. Kang, and S. J. Weisbenner. 2014. ``How University Endowments Respond to Financial Market Shocks: Evidence and Implications.'' \textit{American Economic Review} 104(3): 931--962. \url{https://doi.org/10.1257/aer.104.3.931}.

Loh, C., and E. Buteau. 2017. \textit{A Date Certain: Lessons from Limited Life Foundations}. Cambridge, MA: Center for Effective Philanthropy.

Crook, M. 2023. ``Mission Impossible: Foundation Investment Policy in the Post-COVID World.'' \textit{Journal of Investing} 32(2): 97--107. \url{https://doi.org/10.3905/joi.2022.1.249}.

Deep, A., and P. Frumkin. 2006. ``The Foundation Payout Puzzle.'' In \textit{Taking Philanthropy Seriously: Beyond Noble Intentions to Responsible Giving}, edited by W. Damon and S. Verducci, 189--204. Bloomington: Indiana University Press. Earlier version: Hauser Center Working Paper, 2001.

Dietz, N., B. McKeever, E. Steele, and C. E. Steuerle. 2015. \textit{Foundation Grantmaking over the Economic Cycle}. Washington, DC: Urban Institute.

Ghadamian, R. 2026. ``The Perpetuity Trap: Why Foundation Endowments Misallocate Capital.'' IRSA Working Paper. \url{https://doi.org/10.2139/ssrn.7006300}.

Hengevoss, A., and G. von Schnurbein. 2025. ``All Good Things Come to an End: Emerging Strategic Patterns of Sunset Foundations.'' \textit{Journal of Philanthropy} 30(1): e70015. \url{https://doi.org/10.1002/nvsm.70015}.

Klausner, M. D. 2003. ``When Time Isn't Money: Foundation Payouts and the Time Value of Money.'' Stanford Law School Working Paper. \url{https://doi.org/10.2139/ssrn.445982}.

Internal Revenue Service. 2026. ``Minimum Investment Return.'' Accessed September 2026. \url{https://www.irs.gov/charities-non-profits/private-foundations/minimum-investment-return}.

Ochoa-Carreno, L. 2026. ``Perpetuity and Spend-Down: The Role of Investment and Technological Transitions in Philanthropy.'' \textit{New Political Economy}. Published online March 17, 2026. \url{https://doi.org/10.1080/13563467.2026.2642670}.

Roth Tran, B. 2019. ``Divest, Disregard, or Double Down? Philanthropic Endowment Investments in Objectionable Firms.'' \textit{American Economic Review: Insights} 1(2): 241--256. \url{https://doi.org/10.1257/aeri.20180347}.

Thelin, J. R., and R. W. Trollinger. 2009. \textit{Time Is of the Essence: Foundations and the Policies of Limited Life and Endowment Spend-Down}. Washington, DC: Aspen Institute.

Jaag, C. 2026a. ``Grant or Loan? Allocating Impact Capital under Scarcity.'' Unpublished manuscript.

Jaag, C. 2026b. ``Financing Social Value: Private, Philanthropic, and State Capital.'' Unpublished manuscript.

Jaag, C. 2026c. ``When Future Intervention Weakens Mission: Organizational Commitment under Future Contracting.'' Unpublished manuscript. \url{https://doi.org/10.13140/RG.2.2.14861.50407/2}.

Lo, A. W., E. V. Matveyev, and S. Zeume. 2025. ``The Risk, Reward, and Asset Allocation of Nonprofit Endowment Funds.'' NBER Working Paper 34078, revised February 2026. \url{https://doi.org/10.3386/w34078}.

Mananzala, R., and C. Rojas. 2026. ``The Tidal Foundation: Unmooring Philanthropy From Perpetuity Versus Spend Down.'' Center for Effective Philanthropy, September 10, 2026.

Matveyev, E., and S. Zeume. 2026. ``Financing Nonprofits.'' Working paper, revised July 1, 2026. \url{https://doi.org/10.2139/ssrn.6145866}.

\end{document}